\documentclass[a4paper]{article}
\usepackage[utf8]{inputenc}

\usepackage[top=2.5cm, bottom=2.5cm, left=2.5cm, right=2.5cm]{geometry}

\usepackage{amsmath, amssymb, comment}
\usepackage[hidelinks]{hyperref}

 \usepackage{xcolor}

 \usepackage{amsthm}
 \newtheorem{thm}{Theorem}[section]
 \newtheorem{lm}[thm]{Lemma}
 \newtheorem{res}[thm]{Result}

 \theoremstyle{definition}
 
 \newtheorem{rmk}[thm]{Remark}
 \newtheorem{df}[thm]{Definition}
 \newtheorem{que}[thm]{Question}
 \newtheorem{nt}[thm]{Notation}

 \renewcommand{\phi}{\varphi}

 \newcommand{\FF}{\mathbb F}

 \newcommand{\vspan}[1]{\left \langle #1 \right \rangle}
 \newcommand{\set}[1]{ \left \{ #1 \right \} }
 \newcommand{\sett}[2]{ \left\{ #1 \, \, || \, \, #2 \right \} }

 \newcommand{\pg}{\textnormal{PG}}
 
 \newcommand{\one}{\mathbf 1}
 \newcommand{\zero}{\mathbf 0}

 \newcommand{\ma}{\mathcal A}
 \newcommand{\wt}{\textnormal{wt}}
 \newcommand{\col}{\textnormal{Col}}

\newcommand{\xx}{\mathbf x}
\newcommand{\yy}{\mathbf y}

\newcommand{\pgl}{\textnormal{PGL}}

\title{On additive MDS codes with linear projections}
\author{Sam Adriaensen \thanks{The authors acknowledge the support of PID2020-113082GB-I00 financed by MCIN / AEI / 10.13039/501100011033, the Spanish Ministry of Science and Innovation.} \\ {\it Vrije Universiteit Brussel} \\ \url{sam.adriaensen@vub.be} \and Simeon Ball \footnotemark[1] \\ {\it Universitat Polit\`ecnica de Catalunya} \\ \url{simeon.michael.ball@upc.edu} }
\date{}

\begin{document}

\maketitle

\begin{abstract}
We support some evidence that a long additive MDS code over a finite field must be equivalent to a linear code.
More precisely, let $C$ be an $\FF_q$-linear $(n,q^{hk},n-k+1)_{q^h}$ MDS code over $\FF_{q^h}$.
If $k=3$, $h \in \set{2,3}$, $n > \max \set{q^{h-1},h q -1} + 3$, and $C$ has three coordinates from which its projections are equivalent to linear codes, we prove that $C$ itself is equivalent to a linear code.
If $k>3$, $n > q+k$, and there are two disjoint subsets of coordinates whose combined size is at most $k-2$ from which the projections of $C$ are equivalent to linear codes, we prove that $C$ is equivalent to a code which is linear over a larger field than $\FF_q$.
\end{abstract}

\section{Introduction}

MDS codes, which are codes meeting the Singleton bound, are very useful in different applications of coding theory and cryptography, such as error-correcting codes \cite[Chapter 6]{ball20}, secret sharing \cite{pieprzykzhang}, and distributed storage \cite{dimakis}.
The classical examples are the Reed-Solomon codes.
They are constructed as follows.
Take a finite field $\FF_q$ and denote its elements as $\alpha_1, \dots, \alpha_q$.
Choose an integer $k\leq q$ and for any polynomial 
\(f(X) = \sum_{i=0}^{k-1} f_i X^i \in \FF_q[X] \)
of degree smaller than $k$, define $f(\infty) = f_{k-1}$.
Then
\[
 \sett{(f(\alpha_1), \dots, f(\alpha_q),f(\infty))}{f \in \FF_q[X], \, \deg(f) < k}
\]
is the $k$-dimensional Reed-Solomon code over $\FF_q$.
This is an MDS code of length $q+1$.
If $k \in \set{3,q-1}$ and $q$ is even, this code can be extended to an MDS code of length $q+2$.

It is generally believed that if $C$ is an $(n,q^k,d)_q$ MDS code with $d \geq 3$ and $k < q$, then $n \leq q+1$, with a few known exceptions.
This belief is referred to as the \emph{MDS conjecture}.
For linear MDS codes over $\FF_q$, with $q=p^h$ and $p$ prime, this conjecture is known to hold if one of the following conditions is met:
\begin{itemize}
    \item $q$ is prime,
    \item $q$ is not prime and $k \leq C p^{\lfloor \frac{h+1}2 \rfloor}$, for some constant $C$ that depends on the parity of $q$ and on whether $q$ is a square.
\end{itemize}
For more details, we refer to the recent survey \cite{balllavrauw19}.
The MDS conjecture has been verified for non-linear codes of alphabet size up to 8 \cite{kokkalakrotovostergard,kokkalaostergard}.

In this paper, we study additive MDS codes over finite fields.
It was shown in \cite{gamboa} that these codes are equivalent to certain geometric objects, called pseudo-arcs.
The general link between additive codes over finite fields and certain geometric objects is explained in \S \ref{SecGeom}.
Pseudo-arcs have been investigated in the context of generalised quadrangles, and we revise the most relevant results for the study of MDS codes in \S \ref{SecGen}.
In \S \ref{SecMain}, we prove the following theorems.

\begin{thm}
 \label{ThmMainK=3}
 Let $C$ be an $\FF_q$-linear $(n,q^{3h},n-2)_{q^h}$ MDS code over $\FF_{q^h}$.
 Suppose that one of the following holds.
 \begin{enumerate}
  \item $h=2$, and $n \geq 2q+3$.
  \item $h=3$, and $n \geq q^2 + 3 + \delta_{2,q}$.
 \end{enumerate}
 If $C$ has at least three coordinates from which the projections are equivalent to linear codes, then $C$ itself is equivalent to a linear code.
\end{thm}

\begin{thm}
 \label{ThmMainKLargeIntro}
 Let $C$ be an $\FF_q$-linear $(n,q^{hk},n-k+1)_{q^h}$ MDS code over $\FF_{q^h}$, with $k > 3$.
 Suppose that there exist two subsets $A, B \subset [1,n]$ such that
 \begin{itemize}
  \item $A \cap B = \varnothing$,
  \item $|A| + |B| \leq k-2$,
  \item the projections of $C$ from $A$ and $B$ are equivalent to linear codes.
 \end{itemize}
 If $n > q+k$, then $C$ is equivalent to an $\FF_{q^s}$-linear code for some divisor $s > 1$ of $h$.
 Moreover, if $n > q^e+k$, with $e$ the largest divisor of $h$, which is strictly smaller than $h$, then $C$ is equivalent to a linear code.
\end{thm}

For a slightly stronger version of the latter theorem, see Theorem \ref{ThmMainKLarge}.

\section{Preliminaries}
 \label{SecPrel}

We will start by revising the basics of coding theory.
For a survey, see e.g.\ \cite{ball20}.
Suppose that $A$ is a finite set of size $q$.
A \emph{code} of \emph{length} $n$ over the \emph{alphabet} $A$ is a subset $C$ of $A^n$.
The \emph{(Hamming) distance} between two vectors $\xx = (x_1,\dots,x_n)$ and $\yy = (y_1,\dots,y_n)$ is given by
\[
 d(\xx,\yy) = |\sett{i \in [1,n]}{x_i \neq y_i}|.
\]
The \emph{minimum distance} of $C$ is given by
\[
 d(C) = \min \sett{d(\xx,\yy)}{\xx, \yy \in C, \, \xx \neq \yy}.
\]
If $|C| = M$ and $d(C) = d$, we call $C$ an $(n,M,d)_q$ code.

Suppose that the alphabet $A$ of $C$ is an abelian group.
Then we call $C$ \emph{additive} if $C$ is closed under addition, i.e.\
\[
 (x_1, \dots, x_n), (y_1,\dots,y_n) \in C \Longrightarrow (x_1+y_1, \dots, x_n+y_n) \in C.
\]
Moreover, suppose that the alphabet is a finite field $\FF_{q^h}$.
Then $\FF_{q^h}^n$ is also an $nh$-dimensional vector space over $\FF_q$.
If $C$ is an $\FF_q$-subspace of $\FF_{q^h}^n$ we call $C$ $\FF_q$-\emph{linear}.
An $\FF_{q^h}$-linear code over the alphabet $\FF_{q^h}$ is simply called linear.
Note that if $C$ is $\FF_q$-linear, it is also additive.

The \emph{weight} of a codeword $\xx$ in an additive code is the number of coordinate positions in which $\xx$ has a non-zero entry, and we denote it by $\wt(\xx)$.
Note that $d(\xx,\yy) = \wt(\xx-\yy)$, which implies that
\[
 d(C) = \min \sett{ \wt(\xx) }{ \xx \in C \setminus \set \zero }.
\]

\begin{df}
 \label{DefEq}
 Let $C$ and $D$ be two $\FF_q$-linear codes over $\FF_{q^h}$.
 We call $C$ and $D$ \emph{$\FF_q$-equivalent} (or simply equivalent if $\FF_q$ is clear from context) if one can be obtained from the other by
 \begin{itemize}
  \item permuting the coordinate positions,
  \item in each coordinate position, applying an $\FF_q$-linear automorphism of $\FF_{q^h}$.
 \end{itemize}
 We are allowed to apply different $\FF_q$-linear automorphisms in the different coordinate positions.
\end{df}

Note that equivalent codes have the same size and minimum distance.

One of the most important bounds on the parameters of a code is the \emph{Singleton bound}.
\begin{thm}[Singleton bound]
 Let $C$ be an $(n,M,d)_q$ code.
 Then
 \[
  M \leq q^{n-d+1}.
 \]
\end{thm}
Codes attaining equality in the Singleton bound are called \emph{maximum distance separable} codes, or simply \emph{MDS} codes.
The characteristic property of an $(n,q^k,n-k+1)_q$ MDS code over alphabet $A$ is that if one chooses the entries of a vector of $A^n$ in $k$ positions, the vector can be completed in a unique way to a codeword of $C$.
One of the most central questions in the study of MDS codes is the following.

\begin{que}
Given positive integers $q$, $k$, what is the largest number $n$ such that an MDS code with parameters $(n,q^k,n-k+1)_q$ exists?
\end{que}

Quite some effort has been invested in (partially) answering this question for linear MDS codes, as was mentioned in the introduction.
In this paper, we will investigate $\FF_q$-linear MDS codes over $\FF_{q^h}$.
On one end of the spectrum, we have the case where $q$ is prime, in which case a code is $\FF_q$-linear if and only if it is additive.
On the other end of the spectrum, we have the case where $q=q^h$, in which case a code is $\FF_q$-linear if and only if it is linear.

\section{The geometry behind additive codes}
 \label{SecGeom}

A linear $(n,q^k,d)_q$ code can be linked to a point set in $\pg(k-1,q)$ and vice versa.
In \cite{gamboa}, this link was generalised to additive codes over finite fields.
In this section, we will explain this link and prove some extra properties. First we note that an additive code over $\FF_{q^h}$ must be linear over some subfield, which we fix to be $\FF_q$.

Consider an $\FF_q$-linear $(n,q^k,d)_{q^h}$ code $C$ over $\FF_{q^h}$.
Let $G \in \FF_{q^h}^{k \times n}$ be a matrix whose rows form an $\FF_q$-basis for $C$.
We call $G$ a \emph{generator matrix} of $C$ (over $\FF_q$).
Choose an $\FF_q$-basis $\alpha_1, \dots, \alpha_h$ for $\FF_{q^h}$, and write 
\(
 \boldsymbol \alpha = \begin{pmatrix} \alpha_1 & \dots & \alpha_h \end{pmatrix}^t
\).
Then the $j$\textsuperscript{th} column of $G$ is of the form $G_j \boldsymbol \alpha$ for some $G_j \in \FF_q^{k \times h}$.
Let $\pi_j$ denote the subspace of $\pg(k-1,q)$ corresponding to the column space of $G_j$.
Note that the (projective) dimension of $\pi_j$ is at most $h-1$.

\begin{rmk}
$\pi_j$ is independent of the choice of the $\FF_q$-basis $\alpha_1, \dots, \alpha_h$ of $\FF_{q^h}$.
Indeed, take another $\FF_q$-basis $\beta_1, \dots, \beta_h$ of $\FF_{q^h}$, and write 
\(
\boldsymbol \beta = \begin{pmatrix}
 \beta_1 & \dots & \beta_h
\end{pmatrix}^t
\).
Then
\(
 \boldsymbol \beta = M \boldsymbol \alpha
\)
for some non-singular $M \in \FF_q^{h \times h}$.
Therefore,
\[
 G_j \boldsymbol \alpha = (G_j M^{-1}) \boldsymbol \beta.
\]
Since $M^{-1}$ is non-singular, $G_j$ and $G_j M^{-1}$ have the same column space.
\end{rmk}

The non-zero codewords of $C$ are the vectors $\mathbf a G$, $\mathbf a \in \FF_q^{k} \setminus \set \zero$.
Note that $\mathbf a G$ has a zero in position $j$ if and only if $\mathbf a G_j \boldsymbol \alpha = 0$.
Since the entries of $\boldsymbol \alpha$ are linearly independent over $\FF_q$, this is equivalent to $\mathbf a G_j = \zero$.
This happens if and only if every column of $G_j$ is orthogonal to $\mathbf a$, or equivalently if $\pi_j$ is contained in the hyperplane $a_1 X_1 + \dots + a_k X_k = 0$ of $\pg(k-1,q)$.
Therefore,
\[
 d = \min_{\substack{\text{hyperplane } \Pi \\ \text{of } \pg(k-1,q)}} |\sett{j}{\pi_j \not \subseteq \Pi}|.
\]

\begin{df}
A \emph{projective $h-(n,k,d)_q$ system} is a multiset 
\(
 \set{\pi_1, \dots, \pi_n}
\)
of subspaces in $\pg(k-1,q)$ that have dimension at most $h-1$, span the entire space, and such that
\[
 d = \min_{\substack{\text{hyperplane } \Pi \\ \text{of } \pg(k-1,q)}} |\sett{j}{\pi_j \not \subseteq \Pi}|.
\]
We will also refer to it as a projective $h$-system.
\end{df}

As we just saw, an $\FF_q$-linear $(n,q^k,d)_q$ code over $\FF_{q^h}$ gives rise to a projective $h-(n,k,d)_q$ system.
In fact, it gives rise to several such systems, since we have freedom in choosing the generator matrix $G$.
If $G$ is one generator matrix of $C$ over $\FF_q$, the others are exactly the matrices $M G$, $M \in \text{GL}(k,q)$.
Thus, the set of all projective $h$-systems corresponding to $C$ forms an orbit under $\pgl(k,q)$.

\bigskip

Vice versa, given a projective $h-(n,k,d)_q$-system, we can construct an $\FF_q$-linear $(n,q^k,d)_{q^h}$ code by doing the following.
\begin{itemize} 
 \item Order the elements of the projective $h$-system, and denote the elements in this ordering as $\pi_1, \dots, \pi_n$.
 \item For each $\pi_j$, choose a $k \times h$ matrix $G_j$ whose column space corresponds to $\pi_j$.
 \item Choose an $\FF_q$-basis $\alpha_1, \dots, \alpha_h$ of $\FF_{q^h}$.
 \item Construct the matrix
 \[
  G = \begin{pmatrix}
   G_1 \boldsymbol \alpha & \dots & G_n \boldsymbol \alpha
  \end{pmatrix}.
 \]
 \item Let $C$ be the row space of $G$ over $\FF_q$.
\end{itemize}

We have freedom in choosing the ordering of the subspace, which corresponds to a coordinate permutation in the code.
We are also free to choose any matrix $G_j$ whose column space corresponds to $\pi_j$.
A different choice for $G_j$ corresponds to applying an $\FF_q$-linear map in the $j$\textsuperscript{th} coordinate of $C$.

Thus there is an equivalence between:
 \begin{itemize}
  \item Equivalence classes of $\FF_q$-linear codes over $\FF_{q^h}$ (where the notion of equivalence is taking from Definition \ref{DefEq}).
  \item Orbits of projective $h$-systems under $\pgl(k,q)$.
 \end{itemize}
For the case $h=1$, this correspondence is well-known, see e.g.\ \cite[\S 1.1.2]{tsfasmanvladut}

\begin{df}
 A \emph{pseudo-arc} of $(h-1)$-spaces is a set of $(h-1)$-spaces in $\pg(kh-1,q)$ such that any subset of size $k$ spans the entire space.
 If $h=1$, this is simply called an \emph{arc}.
\end{df}

Recall that if $G$ is a generator matrix of $C$, the codewords of $C$ are of the form $\mathbf a G$, and $\wt(\mathbf a G)$ equals the number of subspaces $\pi_j$ not contained in the hyperplane $a_1 X_1 + \dots + a_k X_k = 0$.
This leads to the following result.

\begin{res}[{\cite{gamboa}}]
 An $\FF_q$-linear code over $\FF_{q^h}$ is MDS if and only if its associated projective $h$-system is a pseudo-arc.
\end{res}

We now describe a way to geometrically recognise $\FF_q$-linear codes which are $\FF_q$-equivalent to linear codes.

\begin{df}
 An \emph{$(h-1)$-spread} of $\pg(kh-1,q)$ is a set of $(h-1)$-spaces in $\pg(kh-1,q)$ such that each point is contained in exactly one of these $(h-1)$-spaces.
\end{df}

An $(h-1)$-spread of $\pg(kh-1,q)$ is equivalent to a set of $h$-dimensional subspaces of $\FF_q^{kh}$ such that every non-zero vector is contained in exactly one of them.
The classical way to construct such a set is by considering $\FF_{q^h}^k$ as a $kh$-dimensional $\FF_q$-vector space.
The $\FF_{q^h}$-subspaces of $\FF_{q^h}^k$ of $\FF_{q^h}$-dimension 1 have $\FF_q$-dimension $h$, and every non-zero vector is contained in exactly one of them.
This gives us a so-called \emph{Desarguesian spread}.
In general, the Desarguesian spreads are the spreads which are $\pgl$-equivalent to this construction \cite[Corollary 3.8]{vandevoorde16}.
Desarguesian spreads can also be characterised by the fact that the set of all subspaces spanned by some elements of the spread together with the natural incidence gives us $\pg(k-1,q^h)$.
For a survey on Desarguesian spreads, we refer the reader to \cite{baderlunardon}.

The previous description can be translated into a more algebraic one.
Fix a primitive element $\omega$ of $\FF_{q^h}$.
Then $1, \omega, \dots, \omega^{h-1}$ is an $\FF_q$-basis for $\FF_{q^h}$.
For each element $\alpha \in \FF_{q^h}$, define the matrix $M(\alpha)$ as the matrix representation of the $\FF_q$-linear map 
\(
\FF_{q^h} \to \FF_{q^h}: x \mapsto \alpha x
\)
with respect to the basis $1, \omega, \dots, \omega^{h-1}$.
Alternatively, if the minimal polynomial of $\omega$ over $\FF_q$ is
\(
X^h - a_{h-1} X^{h-1} - \ldots - a_0,
\)
then
\[
 M(\omega) = \begin{pmatrix}
 0 & & & a_0 \\
 1 & \ddots & & a_1 \\
   & \ddots & 0  & \vdots \\
   & & 1 &  a_{h-1} 
 \end{pmatrix},
\]
$M(\omega^j) = M(\omega)^j$, and $M(0)$ is the zero matrix.
Given a vector $v \in \FF_{q^h}^k$,
\[
 \vspan v _{\FF_{q^h}} = \vspan{v, \omega v, \dots, \omega^{h-1} v }_{\FF_q}.
\]
If we replace each element of $\FF_{q^h}$ by its coordinate vector over $\FF_q$ w.r.t.\ the basis $1, \dots, \omega^{h-1}$, we see that
\[
 \sett{\col \begin{pmatrix} M(x_1) \\ \vdots \\ M(x_k) \end{pmatrix} }{ \begin{pmatrix} x_1 \\ \vdots \\ x_k \end{pmatrix} \in \pg(k-1,q^h) }
\]
is a Desarguesian spread of $(h-1)$-spaces in $\pg(kh-1,q)$, where $\col(M)$ denotes the column space of a matrix $M$.

\begin{lm}
 An $\FF_q$-linear $(n,q^{kh},d)_{q^h}$ code over $\FF_{q^h}$ is $\FF_q$-equivalent to a linear code if and only if its associated projective $h$-system consists of $(h-1)$-spaces is contained in a Desarguesian spread.
\end{lm}

\begin{proof}
First suppose that $C'$ is an $\FF_q$-linear code over $\FF_{q^h}$, equivalent to the linear code $C$.
It suffices to prove that we can construct a projective $h$-system from $C$ contained in a Desarguesian spread.
Take an $\FF_{q^h}$-basis $g_1, \dots, g_k$ of $C$.
Then
\[
 G' = \begin{pmatrix}
  - & g_1 & - \\ - & \omega g_1 & - \\
  & \vdots & \\ 
  - & \omega^{h-1} g_1 & - \\
  & \vdots & \\
  - & g_k & - \\
  & \vdots \\
  - & \omega^{h-1} g_k & - 
 \end{pmatrix}
\]
is a generator matrix for $C$ over $\FF_q$.
Write 
\(
g_i = \begin{pmatrix} g_{i1} & \dots & g_{in} \end{pmatrix}
\) for $i \in [1,k]$, and define
\(
 G_j = \begin{pmatrix}
  M(g_{1j}) \\ \vdots \\ M(g_{kj}).
 \end{pmatrix}
\)
Write 
\(
 \boldsymbol \omega = \begin{pmatrix}
  1 & \omega & \dots & \omega^{h-1}
 \end{pmatrix}^t
\).
Then
\(
 G = \begin{pmatrix} G_1 \boldsymbol \omega & \dots & G_n \boldsymbol \omega \end{pmatrix}
\)
is a generator matrix of $C$ over $\FF_q$, and $\col (G_1), \dots, \col (G_n)$ is contained in a Desarguesian spread.

Vice versa, consider a projective $h$-system of $(h-1)$-spaces $\set{\pi_1, \dots, \pi_n}$ contained in a Desarguesian spread.
Then up to $\pgl$-equivalence, there exist a matrix $G = (g_{ij})$ such that 
\(
\pi_j = \col \begin{pmatrix}
 M(g_{1j}) \\ \vdots \\ M(g_{kj})
\end{pmatrix}
\).
Reversing the above part of the proof, we see that the codes associated to this projective $h$-system are the codes which are $\FF_q$-equivalent to the row space of $G$ over $\FF_{q^h}$.
\end{proof}

Lastly, an important tool in this paper will be projections, both of projective $h$-systems and of codes.
We will give the definitions here, and prove their equivalence.

\begin{df}
 The \emph{projection} of an additive $(n,M,d)_q$ code $C$ from position $i$ is the code
 \[
  \sett{(x_1, \dots, x_{i-1}, x_{i+1}, \dots, x_n)}{(x_1, \dots, x_{i-1}, 0, x_{i+1}, \dots, x_n) \in C}.
 \]
\end{df}

\begin{rmk}
 The above process is sometimes also referred to as \emph{shortening}.
\end{rmk}

Let $\pi$ be an $(m-1)$-space in $\pg(k-1,q)$ and let $U$ be the corresponding $m$-dimensional subspace of $V = \FF_q^k$.
The cosets of $U$ in $\FF_q^k$ form a $(k-m)$-dimensional vector space over $\FF_q$, which we denote by $V / U$.
For every subspace $W$ of $V$, define its \emph{projection} from $U$ as the subspace $W+U$ in $V / U$.
In the projective space $\pg(k-1,q) = \pg(V)$, this yields a map from its subspaces to the subspaces of $\pg(k-m-1,q) = \pg(V/U)$.
We call this map the projection from $\pi$.

A more geometric way to describe this projection is as follows.
Take a $(k-m-1)$-space $\Sigma$ in $\pg(k-1,q)$, disjoint to $\pi$.
Map a subspace $\rho$ of $\pg(k-1,q)$ to the subspace $\vspan{\pi,\rho} \cap \Sigma$.
For any choice of $\Sigma$, this map is $\pgl$-equivalent to the projection from $\pi$ as defined above.

\begin{df}
 Let $\ma = \set{\pi_1, \dots, \pi_n}$ be a projective $h$-system in $\pg(k-1,q)$.
 The \emph{projection} of $\ma$ from $\pi_i$ equals the projective $h$-system consisting of the images of the elements of $\ma \setminus \set {\pi_i}$ under the projection from $\pi_i$.
\end{df}

\begin{lm}
 \label{LmProjDifferentDef}
Let $C$ be an $\FF_q$-linear $(n,q^k,d)_{q^h}$ over $\FF_{q^h}$, and let $\ma = \set{\pi_1, \dots, \pi_n}$ be an associated projective $h$-system (where $\pi_j$ corresponds to position $j$ in the codewords of $C$).
Then the projective $h$-systems associated to the projection of $C$ from position $i$ equal (up to $\pgl$-equivalence) the projection of $\ma$ from $\pi_i$.
\end{lm}

\begin{proof}
We will use vector space notation instead of projective space notation for the projective $h$-systems.
Let $G$ be a generator matrix for $C$.
Fix an $\FF_q$-basis $\alpha_1, \dots, \alpha_h$ of $\FF_{q^h}$.
Write 
\(\boldsymbol{\alpha} = \begin{pmatrix} \alpha_1 & \dots & \alpha_h \end{pmatrix}^t \).
Suppose that
\[
 G = \begin{pmatrix}
  G_1 \boldsymbol \alpha & \dots & G_n \boldsymbol \alpha
 \end{pmatrix}.
\]
Denote the columns of $G_j$ by $G_j^1, \dots, G_j^h$, and denote the dual space of $\FF_q^k$ by $(\FF_q^k)^*$.
Then
\begin{align*}
 C = \sett{ \mathbf a G}{ \mathbf a \in \FF_q^k}
 = \sett{\Big( \begin{pmatrix} f(G_j^1) & \dots & f(G_j^h) \end{pmatrix} \boldsymbol \alpha \Big)_{j \in [1,n] }}{ f \in (\FF_q^k)^*}.
\end{align*}
Then the projection of $C$ from position $i$ equals
\begin{align*}
 C' & = \sett{\Big(\begin{pmatrix} f(G_j^1) & \dots & f(G_j^h) \end{pmatrix} \boldsymbol \alpha \Big)_{j \in [1, n] \setminus \set i}}{ f \in (\FF_q^k)^*, \, \col(G_i) \in \ker f} \\
 & = \sett{\Big(\begin{pmatrix} f \left(\vspan{G_j^1,\col(G_i)}\right)  & \dots & f \left(\vspan{G_j^h,\col(G_i)}\right)\end{pmatrix} \boldsymbol \alpha \Big)_{j \in [1,n] \setminus \set i}}{ f \in (\FF_q^k / \col(G_i))^*}.
\end{align*}
So the projective $h$-system associated to $C'$ in vector space notation equals 
\[ \langle \col(G_2),\col(G_1) \rangle, \dots,  \langle \col(G_n),\col(G_1) \rangle \]
in the vector space $\FF_q^k / \col(G_1)$.
In projective terminology, this equals the projection of $\ma$ from $\pi_i$.
\end{proof}

These ideas can be generalised.
The projection of $C$ from a subset of positions can be obtained by consecutively projecting from each of the positions.
Likewise, the projection of a projective $h$-system can be obtained by consecutively projecting from the elements of this subset.
The above lemma generalises in an obvious way to these more general notions of projection.

\section{Pseudo-arcs and translation generalised quadrangles}
 \label{SecGen}

\begin{df}
 A \emph{generalised quadrangle} of order $(s,t)$ is a point-line incidence geometry such that
 \begin{enumerate}
  \item every line is incident with exactly $s+1$ points ($s \geq 1$),
  \item every point is incident with exactly $t+1$ lines ($t \geq 1$),
  \item given a point $P$ and a line $l$ not incident with $P$, $P$ is collinear with a unique point of $l$.
 \end{enumerate}
\end{df}

An \emph{oval} is an arc of size $q+1$ in $\pg(2,q)$.
Tits described a way to construct a generalised quadrangle $T_2(\mathcal O)$ of order $(q,q)$, from an oval.
A \emph{pseudo-oval} is a pseudo-arc of $(h-1)$-spaces of size $q^h+1$ in $\pg(3h-1,q)$.
As a generalisation of Tits' construction, a generalised quadrangle of order $(q^h,q^h)$ can be constructed from a pseudo-oval.
Generalised quadrangles arising from this construction are characterised by certain properties of their automorphism group, and are called \emph{translation generalised quadrangles}.
We refer to the monograph \cite{thasthasvanmaldeghem} for more information.

Suppose that a pseudo-oval is contained in a Desarguesian spread.
We can interpret the elements of the spread as points in $\pg(2,q^h)$.
Hence, the pseudo-oval corresponds to an oval in $\pg(2,q^h)$.
The generalised quadrangle constructed from the pseudo-oval in $\pg(3h-1,q)$ will be isomorphic to the generalised quadrangle constructed from the oval in $\pg(2,q^h)$.
Therefore, people have been looking for pseudo-ovals not contained in Desarguesian spreads, but none have been found so far.
One could wonder whether there even exist pseudo-ovals not contained in Desarguesian spreads.
There is some evidence to the contrary.

Let $\ma$ be a pseudo-arc of $(h-1)$-spaces in $\pg(kh-1,q)$, of size $n$.
One easily verifies that the projection of $\ma$ from any of its elements is a pseudo-arc of $(h-1)$-spaces of size $n-1$ in $\pg((k-1)h-1,q)$.
The projection of a Desarguesian spread from any of its elements is again a Desarguesian spread.
Hence, if a pseudo-arc is contained in a Desarguesian spread, then so are all its projections.
Conversely, can we infer that a pseudo-arc is contained in a Desarguesian spread if some of its projections are?

The following result proved by Penttila and Van de Voorde \cite{penttilavandevoorde} strengthens a result by Casse, Thas, and Wild \cite{cassethaswild}.
A \emph{conic} in a projective plane $\pg(2,q)$ is a set of points satisfying an irreducible homogeneous quadratic equation.
All conics are projectively equivalent, and Segre \cite{segre55} proved that the largest arcs in $\pg(2,q)$, $q$ odd, are exactly the conics, which have $q+1$ points.

\begin{res}[{\cite{penttilavandevoorde}}]
 Consider a pseudo-arc $\ma$ of $(h-1)$-spaces in $\pg(3h-1,q)$, $q$ odd, whose size is at least the size of the largest arc in $\pg(2,q^h)$ which isn't contained in a conic.
 If the projection of $\ma$ from at least one of its elements is contained in a Desarguesian spread, then $\ma$ itself is contained in a Desarguesian spread.
\end{res}

The case where $q$ is even is much more difficult.

\begin{res}[{\cite{rotteyvandevoorde,thas19}}]
 Consider a pseudo-oval $\ma$ in $\pg(3h-1,q)$, with $q>2$ even, and $h$ prime.
 If all projections of $\ma$ are contained in a Desarguesian spread, then $\ma$ is contained in a Desarguesian spread.
\end{res}

The previous results were motivated by their link to translation generalised quadrangles.
The next result arose in the study of additive MDS codes.

\begin{res}[{\cite{gamboa}}]
Let $\ma$ be a pseudo-arc of lines in $\pg(2k-1,q)$ of size at least $q+k$.
If there is a subset $\mathcal S$ of $\ma$ of size $k+1$ with the property that the projection of $\ma$ from any $(k-2)$-subset of $\mathcal S$ is contained in a Desarguesian spread, then $\ma$ itself is contained in a Desarguesian spread.
\end{res}

Using Lemma \ref{LmProjDifferentDef}, we can interpret these results respectively as follows.

\begin{res}
Let $C$ be an $\FF_q$-linear $(n,q^{3h},n-2)_{q^h}$ MDS code over $\FF_{q^h}$.
Suppose that one of the following properties hold:
\begin{enumerate}
 \item \textnormal{\cite{penttilavandevoorde}} $q$ is odd, $n$ is at least the size of the largest arc in $\pg(2,q^h)$ not contained in a conic, and $C$ has a projection which is equivalent to a linear code;
 \item \textnormal{\cite{rotteyvandevoorde,thas19}} $q>2$ is even, $h$ is prime, $n=q^h+1$, and all projections of $C$ are equivalent to linear codes.
\end{enumerate}
Then $C$ is equivalent to a linear code.
\end{res}

\begin{res}[{\cite{gamboa}}]
 \label{ResGamboaCodes}
Let $C$ be an $\FF_q$-linear $(n,q^{2k},n-k+1)_{q^2}$ MDS code over $\FF_{q^2}$ with $n \geq q+k$.
Suppose that there exist $k+1$ positions such that the projection of $C$ from any $k-2$ of these positions is equivalent to a linear code.
Then $C$ is equivalent to a linear code.
\end{res}

Oberve that Result \ref{ResGamboaCodes} can be recovered from Theorem \ref{ThmMainK=3} substituting $h=2$ in case $n \geq 2q+3$.

\section{Proof of the main theorem}
 \label{SecMain}

The characteristic property of an $(n,q^k,n-k+1)_q$ MDS code $C$ over an alphabet $A$ is that if you fix the entries of a vector in $A^n$ in any $k$ positions, there is a unique way to choose the $n-k$ remaining entries from $A$ which gives you a codeword of $C$.
Thus, an $\FF_q$-linear $(n,q^{kh},n-k+1)_{q^h}$ MDS code $C$ over $\FF_{q^h}$ is of the form $C = \sett{ (\xx,f(\xx)) }{ \xx \in \FF_{q^h}^k }$,
for some function $f: \FF_{q^h}^k \to \FF_{q^h}^{n-k}$.
Take $\xx, \yy \in \FF_{q^h}^k$ and $\alpha, \beta \in \FF_q$.
Since $C$ is $\FF_q$-linear, the codeword
\[
 \alpha (\xx,f(\xx)) + \beta(\yy,f(\yy)) = (\alpha \xx + \beta \yy, \alpha f(\xx) + \beta f(\yy))
\]
is also in $C$.
Hence,
\[
 (\alpha \xx + \beta \yy, \alpha f(\xx) + \beta f(\yy)) = (\alpha \xx + \beta \yy, f(\alpha \xx + \beta \yy)).
\] 
This implies that $f$ is $\FF_q$-linear.
In other words, $f$ is of the form
\[
 f: (x_1,\dots,x_k) \mapsto (f_{k+1,1}(x_1) + \ldots + f_{k+1,k}(x_k),\, \dots \,, f_{n,1}(x_1) + \ldots + f_{n,k}(x_k)),
\]
where each $f_{i,j}:\FF_{q^h} \to \FF_{q^h}$ is a linearised polynomial, i.e.\ it is of the form
\[
 f_{i,j}(X) = \sum_{l=0}^{h-1} a_l X^{q^l},
\]
for some coefficients $a_l \in \FF_{q^h}$.
We will use the notation $f_{i,j}$ in this way throughout this section, with $(i,j) \in [k+1,n] \times [1,k]$.
The fact that $C$ is MDS implies that every $f_{i,j}$ is invertible.

The next result applies to a larger class of $\FF_q$-linear codes than only the MDS codes.

\begin{df}
Let $C$ be an $\FF_q$-linear $(n,q^{hk},d)_{q^h}$ code over $\FF_{q^h}$.
We say that $C$ is in \emph{standard form} if
\(
 C = \sett{ (\xx, f(\xx)) }{ \xx \in \FF_{q^h}^k }
\)
for some $\FF_q$-linear map
\( f: \FF_{q^h}^k \to \FF_{q^h}^{n-k} \) satisfying
\begin{itemize}
 \item $f_{k+1,j}(x_j) = x_j$ for all $j \in [1,k]$,
 \item $f_{i,1}(x_1) = x_1$ for all $i \in [k+1,n]$.
\end{itemize}
\end{df}

\begin{lm}
 \label{LmEqLin}
 Assume that $C=\sett{(\xx,f(\xx))}{\xx \in \FF_{q^h}^k}$ is an $\FF_q$-linear $(n,q^{hk},d)_{q^h}$-code over $\FF_{q^h}$ in standard form.
 $C$ is equivalent to a linear code if and only if there exist an invertible linearised polynomial $g$, such that
 \[
  f_{i,j}(x_j) = g(a_{i,j} g^{-1}(x_j))
 \]
 for all $(i,j) \in [k+2,n] \times [2,k]$, with all $a_{i,j} \in \FF_{q^h}$ constants.
\end{lm}

\begin{proof}
Let $D$ be an $\FF_{q^h}$-linear code equivalent to $C$.
We may suppose without loss of generality that we do not need to permute the coordinate positions of $D$, since such a permutation does not affect the linearity of $D$.
The first $k+1$ positions of $C$ are a $[k+1,k,2]_{q^h}$ MDS code.
Therefore, the same holds for the first $k+1$ positions of $D$.
Consider a generator matrix $(M \; z \; N )$ of $D$ over $\FF_{q^h}$, where $M$ is square, and $z$ a column vector.
Then $M$ must have full rank due to the MDS property.
Therefore, $(I_k \; z' \; N')$, with $z'=M^{-1} z$ and $N' = M^{-1} N$ is also a generator matrix for $D$.
Furthermore, $z'$ cannot have any entries equal to zero, for this would violate the MDS property.

Let $\one$ denote the all-one vector, and let $\Delta$ denote the diagonal matrix with the entries of $z'$ on the diagonal.
Consider the $\FF_{q^h}$-linear code $E$ with generator matrix
\[
 G = (I_k \; \one \; \Delta^{-1} N') = \Delta^{-1} (\Delta \; z' \; N').
\]
Then $(\Delta \; z' \; N')$ is also a generator matrix for $E$, which shows that $E$ is an equivalent code to $D$, without permutation of the coordinate positions.
Thus, $C$ can be constructed from $E$ by applying an $\FF_q$-linear permutation $g_i$ of $\FF_{q^h}$ to each coordinate $i$.
Hence,
\begin{align}
 \label{EqCx}
 C &= \sett{\left(g_1(x_1), \dots, g_k(x_k), \sum_{j=1}^k g_{k+1}(x_j), ..., \sum_{j=1}^k g_i (a_{i,j} x_j), \dots \right)}{(x_1,\dots,x_k) \in \FF_{q^h}^k} \\
 &= \sett{\left(y_1, \dots, y_k, y_1 + \dots + y_k, \dots, y_1 + \sum_{j=2}^k f_{i,j}(y_j), \dots \right)}{(y_1, \dots, y_k) \in \FF_{q^h}^k},
\end{align}
with $(a_{i,j})_{(i,j) \in [k+2,n]\times[1,k]} = \Delta^{-1} N'$.

Let $e_j$ denote the vector in $\FF_{q^h}^k$ with a 1 in position $j$, and zeros everywhere else.
Since $C$ contains the codeword $(e_1, f(e_1)) = (e_1, \one)$, this has to equal the codeword
\[
 (g_1(x_1),0,\dots,0,g_{k+1}(x_1), \dots, g_j(a_{i,1}x_1), \dots),
\]
for some $x_1 \in \FF_{q^h}$
Hence, none of the $a_{i,1}$ are zero.
We can replace $E$ by an equivalent linear code by rescaling the last columns of $G$ in such a way that each $a_{i,1} = 1$.

Now choose $j \in [1,k]$.
Consider the codewords $(\beta e_j, f(\beta e_j))$ of $C$ for all $\beta \in \FF_{q^h}^k$.
These correspond to the codewords of $C$ in (\ref{EqCx}) with $(x_1, \dots, x_k) = x_j e_j$.
Looking in positions $j$ and $k+1$, we see that $g_j(x_j) = g_{k+1}(x_j)$ for each $x_j$.
Write $g = g_{k+1}$.

Now choose \( j > k+1 \) and look at the codeword \( (\beta e_1, f(\beta e_1)) \) in positions 1 and $j$.
Then we see that $g_1(x_1) = g_j(x_1)$ for all $x_1$.
Thus $g_i = g$ for all $i$.
In (\ref{EqCx}), we can replace $(x_1, \dots, x_k)$ by $(g^{-1}(x_1), \dots, g^{-1}(x_k))$, since the latter also runs through all elements of $\FF_{q^h}^k$ if we vary $(x_1, \dots, x_k)$.
Therefore,
\[
 C = \sett{\left( x_1, \dots, x_k, x_1 + \dots + x_k, \dots, x_1 + \sum_{j=2}^k g(a_{i,j} g^{-1}(x_j)), \dots \right)}{ (x_1, \dots, x_k) \in \FF_{q^h}^k},
\]
as claimed.

Now, conversely, suppose that $C$ is in standard form and that $f_{i,j}(x_j) = g(a_{i,j} g^{-1}(x_j))$ for all $(i,j) \in [k+2,n] \times [2,k]$.
Let $D$ be the code equivalent to $C$ obtained by applying the linearised polynomial $g^{-1}$ in every position, i.e.\
\begin{align*}
 D & = \sett{ \left( g^{-1}(x_1), \dots, g^{-1}(x_k), \dots, g^{-1}(x_1) + \sum_{j=2}^k a_{i,j} g^{-1}(x_j), \dots \right) }{(x_1,\dots,x_k) \in \FF_{q^h}^k} \\
 & = \sett{ \left( y_1, \dots, y_k, \dots, y_1 + \sum_{j=2}^k a_{i,j} y_j, \dots \right) }{(y_1,\dots,y_k) \in \FF_{q^h}^k}.
\end{align*}
Then $D$ is clearly linear.
\end{proof}

\begin{nt}
From now on, let $f$ and $g$ denote invertible $\FF_{q}$-linearised polynomials over $\FF_{q^h}$.
Denote
\begin{align*}
 f(X) = \sum_{i=0}^{h-1} f_i X^{q^i}, &&
 f^{-1}(X) = \sum_{i=0}^{h-1} \overline f_i X^{q^i}, &&
 g(X) = \sum_{i=0}^{h-1} g_i X^{q^i}, &&
 g^{-1}(X) = \sum_{i=0}^{h-1} \overline g_i X^{q^i}.
\end{align*}
\end{nt}

\begin{rmk}
 \label{RmkLinCoeff}
 Consider the function \( f(X) \equiv g( a g^{-1}(X)) \), for some non-zero constant $a$.
 Consider an $\FF_q$-linear field automorphism $\varphi: X \mapsto X^{q^e}$.
 Then
 \( f(X) \equiv (g \circ \varphi) ( a^{\varphi^{-1}} (g \circ \varphi)^{-1}(X)) \).
 So we can suppose w.l.o.g.\ that $g_0 \neq 0$.
\end{rmk}

\subsection{The case $k \geq 4$}

\begin{df}
 Let \( n_q(k) \) denote the maximum length of a $k$-dimensional linear MDS code over $\FF_q$.
\end{df}

For bounds on \(n_q(k)\), we refer the reader to \cite{balllavrauw19}.
The most elementary bounds are \( q+1 \leq n_q(k)  \leq q+k-1 \) if $2 \leq k \leq q-1$, and $n_q(k) = k+1$ if $k \geq q$.
Moreover, an $\FF_q$-matrix generates an MDS code over $\FF_q$ if and only if it generates an MDS code over $\FF_{q^e}$ for any integer $e \geq 1$.
Hence, \( n_q(k) \leq n_{q^e}(k) \).

\begin{lm}
 \label{LmSemiLin}
 Let $f$ be an invertible $\FF_q$-linearised polynomial over $\FF_{q^h}$, and take $a \in \FF_{q^h}$.
 Then $f(a f^{-1}(X))$ is linear if and only if $f$ is $\FF_q(a)$-semi-linear.
\end{lm}

\begin{proof}
First suppose that $f$ is $\FF_q(a)$-semi-linear, thus there exists a field automorphism $\sigma$ of $\FF_{q^h}$ such that $f(\alpha X) = \alpha^\sigma f(X)$ for each $\alpha \in \FF_q(a)$.
Note that since $f = g \circ \sigma$ for some $\FF_q(a)$-linearised polynomial $g$, $f^{-1} = \sigma^{-1} \circ g^{-1}$, which implies that \( f^{-1}(\alpha X) = \alpha^{\sigma^{-1}} f^{-1}(X) \) for all $\alpha \in \FF_q(a)$.
Then
\[
 f( a f^{-1}(X)) = f ( f^{-1} (a^\sigma X)) = a^\sigma X,
\]
is linear.
Vice versa, suppose that \( f(a f^{-1}(X)) \equiv b X\) for some $b \in \FF_{q^h}$.
Then
\( a f^{-1}(X) = f^{-1}(b X) \).
This implies that
\[
 \sum_{i=0}^h a \overline f_i X^{q^i}
 = \sum_{i=0}^h b^{q^i} \overline f_i X^{q^i}.
\]
Hence, for every $i$, either $\overline f_i=0$ or $a = b^{q^i}$.
Since $f^{-1} \not \equiv 0$, there exists some $i$ with $a = b^{q^i}$.
Fix the smallest non-negative such $i$.
If $\FF_q(a) = \FF_{q^s}$, then $i < s$ and $\overline f_j \neq 0 \implies j \equiv i \pmod s$.
Thus $f^{-1}(\alpha X) = \alpha^{q^i} f^{-1}(X)$ for each $\alpha \in \FF_{q^s}$, which implies that $f^{-1}$, and hence also $f$, is $\FF_q(a)$-semi-linear.
\end{proof}

\begin{thm}
 \label{ThmMainKLarge}
 Let $C$ be an $\FF_q$-linear $(n,q^{hk},n-k+1)_{q^h}$ MDS code over $\FF_{q^h}$, with $k > 3$.
 Suppose that there exist two subsets $A, B \subset [1,n]$ such that
 \begin{itemize}
  \item $A \cap B = \varnothing$,
  \item $|A| + |B| \leq k-2$,
  \item the projections of $C$ from $A$ and $B$ are equivalent to linear codes.
 \end{itemize}
 If 
 \begin{equation}
  \label{EqMinQ}
  n > |A \cup B| + n_q(k-|A \cup B|), 
 \end{equation}
 then $C$ is equivalent to an $\FF_{q^s}$-linear code for some divisor $s > 1$ of $h$.
 In particular, if 
 \begin{equation}
  \label{EqMinQE}
  n > |A \cup B| + n_{q^e}(k-|A \cup B|), 
 \end{equation}
 with $e$ the largest divisor of $h$ smaller than $h$, then $C$ is equivalent to a linear code.
\end{thm}

\begin{proof}
Since $C$ is an MDS code, we can permute the coordinates of $C$ ad libidum and then put it into standard form without any more coordinate permutations.
Thus, $C$ is equivalent to a code 
\(C' = \sett{(\xx,f(\xx))}{\xx \in \FF_{q^h}^k}\)
in standard form, such that its projections from $A', B' \subseteq [3,k]$, with $A' \cap B' = \varnothing$, are equivalent to linear codes. 
Lemma \ref{LmEqLin} implies that there exist invertible $\FF_q$-linearised polynomials $g_1$ and $g_2$ and constants $(a_{ij})^{i \in [k+1,n]}_{j \in [2,k] \setminus A'}$ and $(b_{ij})^{i \in [k+1,n]}_{j \in [2,k] \setminus B'}$ such that
\begin{enumerate}
 \item $a_{k+1,j} = 1$ and $b_{k+1,j} = 1$,
 \item $f_{i,j}(x_j) = g_1(a_{i,j} g_1^{-1}(x_j))$ for all $j \in [2,k] \setminus A'$,
 \item $f_{i,j}(x_j) = g_2(b_{i,j} g_2^{-1}(x_j))$ for all $j \in [2,k] \setminus B'$.
\end{enumerate}
Now consider the code $D$ equivalent to $C$ obtained by applying $g_1^{-1}$ in every coordinate.
Writing $y_j = g_1^{-1}(x_j)$, we obtain that \( D = \sett{(\yy,f'(\yy))}{\yy\in\FF_{q^h}^k} \) in standard form satisfying:
\begin{enumerate}
 \item \( f'_{i,j}(y_j) = a_{i,j} y_j \) for all \( j \in [2,k] \setminus A' \),
 \item \( f'_{i,j}(y_j) = g(b_{i,j} g^{-1}(y_j)) \) for all \( j \in [2,k] \setminus B' \),
\end{enumerate}
with \( g = g_1^{-1} \circ g_2 \).

Now consider the projection of $D$ from $A' \cup B'$.
It equals the $\FF_{q^h}$-linear code generated by the matrix
\[
 \left( \begin{array}{c c | c c c}
  1 & \mathbf 0 & 1 & \dots & 1 \\
  \mathbf 0 & I_{k-1-|A\cup B|} & & (a_{j,i})^{j\in[2,k]\setminus (A'\cup B')}_{i\in[k+1,n]}
 \end{array} \right).
\]
Note that it is at least 2-dimensional.
If all of the entries of this matrix were in $\FF_q$, it would generate a $(k-|A\cup B|)$-dimensional linear MDS code over $\FF_q$ of length \( n - |A \cup B| > n_q(k-|A \cup B|) \), a contradiction.
Thus, there exists some $a_{i,j} \notin \FF_q$.
Fix these coordinates $(i,j)$.
Since $j \notin B'$,
\(
 g(b_{i,j} g^{-1}(X)) \equiv a_{i,j} X.
\)
Denote \(\FF_q(a_{i,j}) \) by \(\FF_{q^s}\).
Then by Lemma \ref{LmSemiLin}, $g$ is $\FF_{q^s}$-semi-linear.
By Remark \ref{RmkLinCoeff}, we may suppose that $g_0 \neq 0$, from which it follows that $g$ is $\FF_{q^s}$-linear.
This implies that every function \( g(b_{i,j} g^{-1}(y_j)) \) is also $\FF_{q^s}$-linear.
Therefore, every $f'_{i,j}$ is $\FF_{q^h}$-linear if $j \notin A'$ and $\FF_{q^s}$-linear if $j \in A'$.
This implies that $D$ is an $\FF_{q^s}$-linear code.

The last part of the theorem follows from inductively applying the first part, relying on the fact that $n_{q^s}(k)$ is increasing in $s$.
\end{proof}

To make the statement of the proof a bit less technical, one can replace the right-hand sides of the bounds in (\ref{EqMinQ}) and (\ref{EqMinQE}) respectively by $q+k$ and $q^e+k$, as was done in Theorem \ref{ThmMainKLargeIntro}.

One could wonder how tight the bound on $n$ is in the above theorem.
Let us illustrate this with an example where $k=4$, $q \geq 5$, and \( |A| = |B| = 1\).
Write $n = n_q(4)$.
Note that $n_q(4) \geq 6$ since $q \geq 5$.
Take a linear $[n,4,n-3]$ MDS code over $\FF_q$, and choose a generator matrix $G$ of this code in standard form.
Write
\[
 G = \begin{pmatrix}
 1 & & & & 1 & 1 & \dots & 1 \\
 & 1 & & & 1 & a_{6,2} & \dots & a_{n,2} \\
 & & 1 & & 1 & a_{6,3} & \dots & a_{n,3} \\
 & & & 1 & 1 & a_{6,4} & \dots & a_{n,4}
\end{pmatrix}
\]

Consider the linear code $C$ generated by $G$ over $\FF_{q^h}$.
This code is still MDS.
We want to make a slight modification so that the code still has two postions from which its projection is equivalent to a linear code, but so that the code itself is no longer equivalent to a linear code.

Choose $\alpha, \beta \in \FF_{q^h}\setminus\FF_q$ with \( \FF_q \varsubsetneq \FF_q(\alpha) \cap \FF_q(\beta) \).
Note that if we replace $a_{n,3}$ by $\alpha$, $G$ still generates an MDS code over $\FF_{q^h}$.
We make one more modification.
Choose an invertible $\FF_q$-linearised polynomial $g \in \FF_{q^h}[X]$ that is not semi-linear over \( \FF_q(\alpha) \cap \FF_q(\beta) \).
Now consider the code $C$ consisting of the codewords
\[
 \Big( x_1,x_2,x_3,x_4,x_1+x_2+x_3+x_4, \underbrace{\dots, x_1 + \sum_{j=2}^4 a_{i,j} x_j \dots}_{i=6,\dots,n-1}, x_1 + a_{n,2} x_2 + \alpha x_3 +  g(\beta g^{-1}(x_4)) \Big)
\]
where $(x_1,x_2,x_3,x_4)$ runs over $\FF_{q^h}^4$.

\begin{lm}
\emph{The projections of $C$ from the third and from the fourth coordinate are equivalent to linear codes, but $C$ is not equivalent to a linear code.}
\end{lm}

\begin{proof}
The projection from the fourth coordinate being linear is obvious.
Now project from the third coordinate.
Apply $g^{-1}$ in every coordinate position, and write $y_j = g^{-1}(x_j)$.
Then the obtained code is
\[
 \Big( y_1,y_2,y_4,y_1+y_2+y_4, \dots, y_1 + \sum_{j=2,4} a_{i,j} y_j \dots, y_1 + a_{n,2} y_2 + \beta y_4 \Big)
\]
which is clearly linear.

Now suppose that $C$ is equivalent to a linear code.
By Lemma \ref{LmEqLin}, there exists an invertible $\FF_q$-linearised polynomial $f$ such that
\begin{enumerate}
 \item \( f(\gamma f^{-1}(X)) \equiv \alpha X \), for some \( \gamma \in \FF_{q^h} \).
 This implies that \( f^{-1} ( \alpha f(X)) \equiv \gamma X \).
 Hence, $f$ is $\FF_q(\alpha)$-semi-linear.
 \item \( f(\delta f^{-1}(X)) \equiv g( \beta g^{-1}(X)) \) for some \( \delta \in \FF_{q^h} \).
 This implies that \( (f^{-1} \circ g) ( \beta (f^{-1} \circ g)(X)) \equiv \delta X \).
 Hence, $f^{-1} \circ g$ is $\FF_q(\beta)$-semi-linear.
\end{enumerate}
This implies that $g \equiv f \circ (f^{-1} \circ g)$ is semi-linear over \( \FF_q(\alpha) \cap \FF_q(\beta) \), which contradicts the way we chose $g$.
\end{proof}

The question remains how we must choose $\alpha$, $\beta$, and $g$ such that $C$ is an MDS code.
It must hold that if we choose 4 coordinate positions of $C$, and we want the entries of the positions to be zero, this is only possible if all $x_j=0$.
Since $C$ is a slight modification of an MDS code, the only thing we need to check is the following.
Choose 4 coordinate positions of $C$, including the last one.
Write down the linear system that says we have zeros in these positions.
\[
 \begin{cases}
 \sum_{j=1}^4 a_{i,j} x_{i,j} = 0 & i=i_1,i_2,i_3 \\
 x_1 + a_{n,2} x_2 + \alpha x_3 + g( \beta g^{-1} (x_4)) = 0
 \end{cases}
\]
We only need to prove that if we eliminate the variables $x_1, x_2, x_3$ from these equations, we are left with an equation that implies that $x_4=0$.
Then $x_1=x_2=x_3=0$ by the MDS property of our original code.
After elimination, the remaining equation is of the form
\[
 g( \beta g^{-1} (x_4)) = (\lambda_1 \alpha + \lambda_2) x_4
\]
for some \( \lambda_1, \lambda_2 \in \FF_q \).
In other words, if no non-zero $x_4$ satisfies $g(\beta g^{-1}(x_4)) \in \vspan{1,\alpha}_{\FF_q}$, then $C$ is an MDS code.

Although it is not evident to find a general construction for $\alpha$, $\beta$, and $g$, for small values of $q$, we found examples by computer.

\subsection{The case $k=3$}

The case $k=3$ is more difficult, since then there is not enough ``overlap'' between the different projections.

\begin{df}
 \label{DfProp}
 Let $f$ and $g$ be invertible $\FF_q$-linearised polynomials over $\FF_{q^h}$.
 We say that $(f,g)$ satisfies property $(Prop_m)$ if there exist triples $(a_j,b_j,c_j) \in (\FF_{q^h}^*)^3$, $1 \leq j \leq m$, such that
 \[
 a_j f( b_j f^{-1} (X)) \equiv g (c_j g^{-1} (X))
\]
 and for every $i \neq j$, it holds that $a_i \neq a_j$, $b_i \neq b_j$, and $c_i \neq c_j$.
\end{df}
We will always suppose that $a_1 = b_1 = c_1 = 1$.

\begin{lm}
 \label{LmProp}
 Suppose that for a given value of $m$, the only pairs of $\FF_q$-linearised polynomials of $\FF_{q^h}$ satisfying property $(Prop_m)$ are monomials.
 Then every $\FF_q$-linear $(n,q^{3h},n-2)_{q^h}$ code $C$, with $n \geq 3 + m$ which has at least three projections that are equivalent to linear codes, is itself equivalent to a linear code.
\end{lm}

\begin{proof}
Suppose the hypothesis from the lemma holds.
As in the previous proof, we may suppose that $C$ is in standard form, and the projections from the first three coordinates are equivalent to linear codes.
Since the projections from the second and third coordinate are equivalent to linear codes, Lemma \ref{LmEqLin} 
implies that $C$ is of the form
\[
 \sett{(x,y,z,\dots, x + g_1(a_j g_1^{-1}(y)) + g_2(b_j g_2^{-1}(z)), \dots )  }{(x,y,z) \in \FF_{q^h}^3},
\]
with $g_1$ and $g_2$ invertible $\FF_q$-linearised polynomials, and $a_j$, $b_j$ constants.
This code is equivalent to
\begin{align*}
 C' =& \set{(g_1^{-1}(x),g_1^{-1}(y),g_1^{-1}(z),\dots, g_1^{-1}[ x + g_1(a_j g_1^{-1}(y)) + g_2(b_j g_2^{-1}(z)) ], \dots )  } \\
 = &\set{(x,y,z,\dots, x + a_j y + g_1^{-1} \circ g_2(b_j (g_1^{-1} \circ g_2)^{-1}(z)), \dots )  }.
\end{align*}
Now write $g = g_1^{-1} \circ g_2$.
The fact that the projection from the first coordinate is equivalent to a linear code implies that there exists an invertible $\FF_q$-linearised polynomial $f$ such that
\[
 a_j^{-1} g (b_j g^{-1} (z) ) = f( c_j f^{-1} (z))
\]
for some constants $c_j$.
The fact that $C$ is an MDS code implies that $a_i \neq a_j$, $b_i \neq b_j$, and $c_i \neq c_j$ if $i\neq j$.
Thus, $(f,g)$ satisfies property $(Prop_m)$.
Therefore, they are monomials, which implies that $g(b_j g^{-1}(X))$ is a linear function for every value $b_j$.
Hence, $C'$ is linear.
\end{proof}

\begin{df}
The \emph{Dickson matrix} of an $\FF_q$-linearised polynomial $F(X) = \sum_i F_i X^{q^i}$ over $\FF_{q^h}$ is the matrix
\[
 M_F = \begin{pmatrix}
 F_0 & F_1 & \dots & F_{h-1} \\
 F_{h-1}^q & F_0^q &  & F_{h-2}^q \\
  & & \ddots \\
 F_1^{q^{h-1}} & F_2^{q^{h-1}} & \dots & F_0^{q^{h-1}}
 \end{pmatrix}
 = \left(F_{j-i}^{q^i} \right)_{i,j \in [0,h-1]},
\]
where we take the indices of $F_{j-i}$ modulo $h$.
\end{df}

For more on Dickson matrices, see e.g.\ \cite{wuliu}.
An important property is that $M_{F \circ G} = M_F M_G$ and hence $M_{F^{-1}} = M_F^{-1}$.

\begin{lm}
 \label{LmAtLeastOneZeroCoeff}
 Suppose that $m > \max\{q^{h-1}, hq-1\}$.
 If $(f,g)$ satisfies property $(Prop_m)$, then \[ |\sett{i}{f_i = 0}| = |\sett{i}{g_i = 0}| \geq 1.\]
\end{lm}

\begin{proof}
By Remark \ref{RmkLinCoeff}, we may suppose w.l.o.g.\ that $f_0 \neq 0$ and $g_0 \neq 0$.
Let $(a_j,b_j,c_j)$ be as in Definition \ref{DfProp}.
Write $B_{i,j} = b_j^{q^i} - b_j$ for $i = 1, \dots, h-1$.
Then
\begin{align*}
 f(b_j f^{-1}(X)) \equiv \sum_{i=0}^{h-1} f_i \left( b_j \sum_{l=0}^{h-1} \overline f_l X^{q^l} \right)^{q^i}
 \equiv \sum_{i=0}^{h-1} b_j^{q^i} f_i \sum_{l=0}^{h-1} \overline f_{l-i}^{q^i} X^{q^l} 
 \equiv b_j F(X) + \sum_{i=1}^{h-1} B_{ij} f_i \sum_{l=0}^{h-1} \overline f_{l-i}^{q^i} X^{q^l} 
\end{align*}
for some $\FF_q$-linearised polynomial $F$, where we take the indices of $\overline f_i$ modulo $h$.
Since $b_1=1$, it holds that $f( b_1 f^{-1}(X)) \equiv X$ and all $B_{i,1}=0$, thus $F(X) = X$.

Likewise define $C_{ij} = c_j^{q^i} - c_j$.
Then for every $j$,
\begin{align}
 \label{EqCoeffFAndG}
 a_j \left( b_j X + \sum_{l=0}^{h-1} \left( \sum_{i=1}^{h-1} B_{i,j} f_i \overline f_{l-i}^{q^i} \right) X^{q^l} \right)
 \equiv c_j X + \sum_{l=0}^{h-1} \left( \sum_{i=1}^{h-1} C_{i,j} g_i \overline g_{l-i}^{q^i} \right) X^{q^l}.
\end{align}
Consider the Dickson matrices $M_{f^{-1}}$ and $M_f$ of $f^{-1}$ and $f$.
Since $M_f M_{f^{-1}} = I$,
\[
 \begin{pmatrix}
 f_0 & \cdots & f_{h-1}
 \end{pmatrix}
 M_{f^{-1}} = \begin{pmatrix}
 1 & 0 & \cdots & 0
 \end{pmatrix}
\]
Let $\hat M_f$ denote the submatrix of $M_{f^{-1}}^t$ obtained by deleting the top row and first column.
By Cramer's rule, $ \displaystyle f_0 = \frac{ \det \hat M_f }{ \det M_{f^{-1}} }$.
Since we assumed $f_0 \neq 0$, this implies that $\det \hat M_f \neq 0$.
Define the diagonal matrix
\[
 D_f = \begin{pmatrix}
 f_1 \\ & \ddots \\ & & f_{h-1}
\end{pmatrix}
\]
Analogously define $\hat M_g$ and $D_g$.
Write $B_j = \begin{pmatrix} B_{1,j} \\ \vdots \\ B_{h-1,j} \end{pmatrix}$ and $C_j = \begin{pmatrix} C_{1,j} \\ \vdots \\ C_{h-1, j} \end{pmatrix}$.
Looking in (\ref{EqCoeffFAndG}) at the coefficients of $X^{q}, \dots, X^{q^{h-1}}$, we see that for any $j$,
\begin{align}
 \label{EqMatrixForm}
 a_j \hat M_f D_f B_j = \hat M_g D_g C_j.
\end{align}
Take a vector $u$ in the left-kernel of $\hat M_f D_f$.
Then for any $j$, $u \hat M_g D_g C_j = 0$.
Since all entries of $C_{j}$ are polynomial in $c_j$ of degree at most $q^{h-1}$, we find a polynomial equation in $c_j$ of degree at most $q^{h-1}$.
Since $(f,g)$ satisfies $(Prop_{q^{h-1}+1})$, this equation has at least $q^{h-1}+1$ roots, and therefore identically equals the zero polynomial.
This implies that $u \hat M_g D_g = 0$.
We can repeat this argument with the roles of $f$ and $g$ reversed.
It follows that $\hat M_f D_f$ and $\hat M_g D_g$ have the same left kernel.
Since $\hat M_f$ and $\hat M_g$ have full rank, this implies that $D_f$ and $D_g$ have the same rank.
Note that the rank defect of $D_f$ and $D_g$ equal respectively $|\sett i {f_i = 0}|$ and $|\sett i {g_i = 0}|$.

What is left to prove, is that $D_f$ and $D_g$ do not have full rank.
So suppose the contrary.
Consider the field automorphism $\phi: x \mapsto x^q$.
For a matrix $A$, we let $A^\phi$ denote the matrix obtained by applying $\phi$ to all its entries.
Note that 
\[
 B_{i,j}^q = b_j^{q+1} - b_j^q = \begin{cases}
  B_{i+1, j} - B_{1,j} & \text{if } i < h-1, \\
  - B_{1,j} & \text{if } i = h-1.
 \end{cases}
\]
Define the matrix
\[
 L = \begin{pmatrix}
 -1 & 1 \\
 \vdots & & \ddots \\
 -1 & & & 1 \\
 -1
 \end{pmatrix}
\]
Then $B_j^\varphi = L B_j$ and it is easy to check that $L$ has full rank.
If we apply $\phi$ to (\ref{EqMatrixForm}), this yields
\[
 a_j^q \hat M_f^\phi D_f^\phi L B_j = \hat M_g^\phi D_g^\phi L C_j
\]
This implies that
\[
 C_j = a_j D_g^{-1} \hat M_g^{-1} \hat M_f D_f B_j = a_j^q L^{-1} D_g^{-\phi} \hat M_g^{-\phi} \hat M_f^{\phi} D_f^{\phi} L B_j.
\]
If $b_j \notin \FF_q$, then this implies that the non-zero vector $B_j$ is in the kernel of
\[
 a_j D_g^{-1} \hat M_g^{-1} \hat M_f D_f - a_j^q L^{-1} D_g^{-\phi} \hat M_g^{-\phi} \hat M_f^{\phi} D_f^{\phi} L.
\]
The determinant of this matrix is a polynomial in $a_j$ of degree at most $(h-1)q$.
Since it has at least $m - (q-1) > (h-1)q$ zeros, it must be identically zero.
Looking at the coefficient of $a_j^{h-1}$, this implies that $\det(D_g^{-1} \hat M_g^{-1} \hat M_f D_f) = 0$.
But by our assumption, this was a product of full-rank matrices, which yields a contradiction.
\end{proof}

\begin{lm}
 \label{LmInverse}
 If $(f,g)$ satisfies property $(Prop_m)$, then so do $(f^{-1},f^{-1} \circ g)$ and $(g^{-1}, g^{-1} \circ f)$.
\end{lm}

\begin{proof}
 \begin{align*}
  a_j f ( b_j f^{-1} (X)) &\equiv g (c_j g^{-1} (X)) \\
 \iff  f (b_j^{-1} f^{-1} (a_j^{-1} X)) \equiv [ a_j f ( b_j f^{-1} (X))]^{-1} &\equiv 
 [g (c_j g^{-1} (X))]^{-1} \equiv g( c_j^{-1} g^{-1} (X)) \\
\iff b_j^{-1}  f^{-1} (a_j^{-1} f(Y)) &\equiv f^{-1} \circ g (c_j^{-1} g^{-1} \circ f(Y))
 \end{align*}
The last equality follows from the substitution $X = f(Y)$.
Since $(f^{-1} \circ g)^{-1} = g^{-1} \circ f$, it follows that $(f^{-1}, f^{-1} \circ g)$ satisfies property $(Prop_m)$.
The proof for $(g^{-1}, g^{-1} \circ f)$ is completely analogous.
\end{proof}

\begin{lm}
 \label{LmTwoNonZeroCoeff}
 Suppose that $f$ has exactly two non-zero coefficients, and that $f$ is not semi-linear over some subfield $\FF_{q^s} \supset \FF_q$ of $\FF_{q^h}$.
 Then all coefficients of $f^{-1}$ are non-zero.
\end{lm}

\begin{proof}
We can again compose $f$ with a monomial to make $f_0$ one of the non-zero coefficients.
Then the other non-zero coefficient is $f_j$ where $j$ is coprime with $h$.

Let $M$ be the transposition of the Dickson matrix of $f$, that is
\[
 M = M_f^t = \begin{pmatrix}
  f_0  & & & & f_j^{q^{h-j}} \\
  & f_0^q  & & & & \ddots \\
  & & & & & & f_j^{q^{h-1}}         \\ 
  & & & \ddots         \\
  f_j    \\
  & f_j^q \\
  & & \ddots \\
 & & & f_j^{q^{h-j-1}} & & & f_0^{q^{h-1}}
 \end{pmatrix}
\]
Then $M M_{f^{-1}}^t = I_h$.
Looking at the first column, this implies that
\[
 M \begin{pmatrix}
 \overline f_0 \\ \overline f_1 \\ \vdots \\ \overline f_{h-1}
 \end{pmatrix}
 = \begin{pmatrix}
 1 \\ 0 \\ \vdots \\ 0
 \end{pmatrix}.
\]
Thus, by Cramer's rule, we need to prove that if we remove the top row of $M$, then any $(h-1) \times (h-1)$ submatrix of has non-zero determinant.
We do this using the following expression for the determinant of a $n \times n$-matrix $A=(a_{il})$.
\[ \det(A) = \sum_{\sigma \in S_n} \text{sgn}(\sigma)  \prod_{i=1}^n a_{\sigma(i) \, i}, \]
where $S_n$ denotes the symmetric group on $n$ elements and $\text{sgn}(\sigma)$ denotes the sign of the permutation $\sigma$.
We will prove for every $(h-1) \times (h-1)$-submatrix described above that when expressing its determinant in this way, the sum has exactly one non-zero term, hence is not zero.

So remove column $k$ and the top row from $M$, and index the rows and columns in the corresponding matrix by $[1,h-1] \times [0,h-1] \setminus \set k$.
We will denote the indices modulo $h$.
We need to prove that there is a unique bijection $\sigma: [0,h-1] \setminus \set k \to [1,h-1]$ such that $M_{\sigma(i),i} \neq 0$ for all $i$.
Equivalently, for every $i$, $\sigma(i) \in \set{i,i+j} \setminus \set 0$.
Suppose that $\sigma(i) = i$.
If $i-j \neq k$, then $\sigma(i-j) = i-j$, since $\sigma(i-j) = i$ would violate the injectivity of $\sigma$.
Likewise if $\sigma(i) = i+j$ and $i+j \neq k$, then $\sigma(i+j) = i+2j$.

Now write $i \prec l$ if $i \equiv a j \pmod h$ and $l \equiv b j \pmod h$ for some $0 \leq a < b < h$.
In other words, $i \prec l$ if given the residue classes of $i j^{-1}$ and $l j^{-1}$ modulo $h$, the former has the smallest representative in $[0,h-1]$.
Note that this ordering is well-defined because $j$ is coprime with $h$.

Consider a column $i$.
\begin{itemize}
 \item If $i \prec k$, we can walk from column 0 to column $i$ with step size $j$ without encountering column $k$.
 Since $\sigma(0)$ must equal $j$, this means that $\sigma(i)$ must equal $i+j$.
 \item If $i \succ k$, we can walk from column $-j$ to column $i$ with step size $-j$.
 Since $\sigma(-j)$ must equal $-j$, this implies that $\sigma(i)$ must equal $i$.
\end{itemize}
Hence, the only possible way to construct $\sigma$ is as follows.
\[
 \sigma(i) = \begin{cases}
  i+j & \text{if } i \prec k, \\
  i & \text{if } i \succ k.
 \end{cases}
\]
One can easily check that this gives indeed the desired bijection.
\end{proof}

\begin{thm}
 Let $C$ be an $\FF_q$-linear $(n,q^{3h},n-2)_{q^h}$ MDS code over $\FF_{q^h}$.
 Suppose that one of the following holds.
 \begin{enumerate}
  \item $h=2$, and $n \geq 2q+3$.
  \item $h=3$, and $n \geq q^2 + 3 + \delta_{2,q}$.
 \end{enumerate}
 If $C$ has at least three coordinates from which the projections are equivalent to linear codes, then $C$ itself is equivalent to a linear code.
\end{thm}

\begin{proof}
By Lemma \ref{LmProp}, it suffices to prove that if $(f,g)$ satisfies property $(Prop_{n-3})$, then $f$ and $g$ are monomials.
By Lemma \ref{LmAtLeastOneZeroCoeff} and \ref{LmInverse}, we know that $f$ and $f^{-1}$ have at least one coefficient equal to zero.
For $h=2$, this proves that $f$ is a monomial.
Now suppose that $h=3$.
If $f$ would have exactly one coefficient equal to zero, then $f^{-1}$ would have no coefficients equal to zero by Lemma \ref{LmTwoNonZeroCoeff}.
Thus, $f$ has at least two coefficients equal to zero, which implies that $f$ is a monomial.
Likewise for $g$.
\end{proof}

\bigskip

\noindent \textbf{Acknowledgements.}
The authors would like to thank Geertrui Van de Voorde for helpful discussions.

\bibliographystyle{alpha}

\bibliography{ref}

\end{document}